%% file: Allerton09_main.tex
\documentclass[letterpaper, 10pt,twocolumn,conference]{ieeetran}
\IEEEoverridecommandlockouts

\usepackage{amssymb}
\usepackage{fullpage}
\usepackage{times}
\usepackage{amsmath}
\usepackage{amsfonts}
\usepackage{amssymb,dsfont}
\usepackage{latexsym}
\usepackage{xcolor}
\usepackage{pifont}
\usepackage{verbatim}
\usepackage{graphicx}
\usepackage{subfigure}
\title{Rateless Codes for Single-Server Streaming to Diverse Users}

\author{
\authorblockN{Yao Li}
\authorblockA{ECE Department, Rutgers University\\
Piscataway NJ 08854\\
yaoli@winlab.rutgers.edu}
\and
\authorblockN{Emina Soljanin}
\authorblockA{Bell Labs, Alcatel-Lucent\\
Murray Hill NJ 07974, USA\\
emina@alcatel-lucent.com}
\thanks{
This work was supported by the NSF grant No.~CNS 0721888.}
 }

\begin{document}
\maketitle
\begin{abstract}
We investigate the performance of rateless codes for single-server
streaming to diverse users,  assuming that diversity in users is
present not only because they have different channel conditions, but
also because they demand different amounts of information and have
different decoding capabilities. The LT encoding scheme is employed.
While some users accept output symbols of all degrees and decode
using belief propagation, others only collect degree-1 output
symbols and run no decoding algorithm. We propose several
performance measures, and optimize the performance of the rateless
code used at the server through the design of the code degree
distribution. Optimization problems are formulated for the
asymptotic regime and solved as linear programming problems.
Optimized performance shows great improvement in total bandwidth
consumption over using the conventional ideal soliton distribution,
or simply sending separately encoded streams to different types of
user nodes. Simulation experiments confirm the usability of the
    optimization results obtained for the asymptotic regime as a guideline for finite-length code design.
\end{abstract}
\section{Introduction}
\subsection{Motivation}

\input{Allerton09_motivation.tex}

\subsection{Related Work}
\input{Allerton09_relatedwork.tex}

\section{System Model: Multicast Over BEC Channels} \label{sec:model}
\input{Allerton09_multicastmodel.tex}

\section{The Optimization Problem in the Asymptotic Regime}  \label{sec:theory}
\input{Allerton09_theorybasis2.tex}

\section{Performance Measures and Their Optimization Problem Statements}    \label{sec:probs}
\input{Allerton09_optprob2.tex}

\section{Optimization Results}      \label{sec:results}
\input{Allerton09_results.tex}

\section{Finite-Length Simulation}   \label{sec:simulation}
\input{Allerton09_simulation.tex}
\newpage
\section{Concluding Remarks}
    In this work, we have investigated the performance of LT rateless codes for streaming from a single server to diverse
    users. The degree distributions of the LT-output
    symbols have been optimized
    according to network parameters.
    The degree distribution optimization problems have been
    formulated in the asymptotic regime and solved numerically, and simulations have been conducted to confirm the usability of the
    asymptotic results as a guideline for finite-length code design. The impact of diversity in channel conditions, non-uniform demands
and coding methods of users on transmission latency, channel
utilization and throughput have also been shown through the
optimization results. As demonstrated in Section \ref{sec:results},
following our scheme, the total bandwidth consumption for satisfying
diverse users is considerably reduced compared to either sending
separate streams for different users or sending a stream that is
optimized for only one of the users.
\bibliographystyle{unsrt}
\bibliography{ftrefs}

\end{document}

%% file: Allerton09_motivation.tex
Growing popularity of ubiquitous computing, along
with the surging demand for digital media distribution services such
as YouTube\texttrademark, has brought up the issue of efficient
media sharing in a heterogenous network composed of links of
diverse quality as well as terminals of varied computing power and
demand of media quality.

Consider the air broadcast of digital TV streams. A specialized
``plugged'' receptor, such as an HDTV set at home, may have more
computing power than a small portable device, such as a cellphone,
and hence the former might be able to perform more complex decoding
algorithms than the latter. Meanwhile, the quality of the broadcast
channels may vary due to the location of the receiver, indoors or
outdoors, near or far from the transmitting tower. Moreover, devices
may need different amounts of data to display a video stream
according to screen resolutions.

Here, we are interested in finding some efficient and yet fair way
to provide multicast streaming service to all or a majority of the
receivers bearing such heterogeneity. One straightforward solution
is to transmit separately encoded data streams suitable for
different devices and channels simultaneously, but this requires
extra bandwidth and is hence less than efficient.

Rateless codes \cite{luby,amin} are, roughly speaking, designed for erasure channels in
a way that the set of information symbols may be recovered from any
subset of the encoding symbols of size equal or slightly larger than
that of the information symbol set by simple decoding. The first practical rateless
codes, LT codes, were invented by Michael Luby and published in
2002 \cite{luby}. Another class of rateless codes are Raptor codes, a
version of which has been written into the 3GPP standard for Multimedia Broadcast/Multicast Service \cite{3GPP}.

Rateless codes have the nice features of requiring minimal feedback
from the receiver to the sender and operating well over a range of
channel conditions. These features are particularly suitable for the
broadcast/multicast scenarios.  We investigate the possibility of
simultaneously serving data sinks of highly heterogenous decoding
capabilities and non-uniform demand of information on channels of
diverse quality, with a single rateless coded multicast stream from
the source. Specifically, we study the design of LT codes for the
multicast streaming purpose.

%% file: Allerton09_relatedwork.tex
The performance of LT codes is determined by the degree distribution
of encoding/output symbols. In \cite{luby} and \cite{amin}, the
ideal soliton and robust soliton degree distributions have been
proposed for minimizing the overhead necessary for recovering all
input symbols. However, using these degree distributions when the
number of output symbols collected by the receiver is smaller than
the total number of the input symbols results in recovery of few
input symbols. In \cite{sanghavi}, the optimal degree distributions
for recovering a constant fraction of the input symbols from the
smallest number of output symbols have been studied.

Our work considers multicast streaming to all user nodes with a
single data stream. We deal with simultaneous multiple
heterogeneities such as link diversity, difference in coding
capabilities (e.g., due to limitations in computing resources), and
difference volume of information demand (e.g., low or high
resolution video). We are interested in performance measures
reflecting the collective properties of all the sink nodes of
interest, such as maximum and average latency. Our approach by
designing

Our paper is organized as follows: Section \ref{sec:model}
introduces the system model for the heterogeneous multicasting
network. Section \ref{sec:theory} outlines the guidelines for our
optimization problems in the asymptotic regime. Section
\ref{sec:probs} proposes several performance measures and states the
corresponding optimization problems. Section \ref{sec:results}
presents the optimization results of the problems formulated in
section \ref{sec:probs}. Section \ref{sec:simulation} contains finite-length
simulation results.

%% file: Allerton09_multicastmodel.tex
We consider a streaming network consisting of a single server
(source node) and $n$ users (sink nodes) each directly connected to
the server by a BEC channel, as shown in Figure
~\ref{fig:broadcastmodel}.
\begin{figure}[h]
\centering
\includegraphics[scale=0.3]{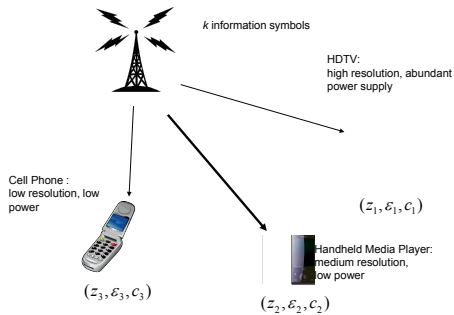}
\vspace{-.3in} \caption{Broadcast/Multicast System
Model}\label{fig:broadcastmodel}
\end{figure}
The source holds $k$ information symbols and broadcasts a rateless
coded stream to all $n$ sinks. The rateless encoder is an LT
encoder\cite{luby} with degree distribution with moment generating
function
\begin{equation}
P^{(k)}(x)=p_1^{(k)}x+p_2^{(k)}x^2+\dots+p_k^{(k)}x^k.
\label{eq:mgf}
\end{equation}
The LT encoder generates potentially an infinite number of output
symbols and broadcast the output stream along all BEC links.

There are two types of sink nodes which differ in the way the LT
code is decoded. One type of sinks use the belief propagation (BP)
algorithm \cite{luby} to recover the input symbols from the received
output symbols, while the other type of sinks only accept and
recover information from degree-1 output symbols received from the
source. The first type are referred to as decoding, and the second
as non-decoding sinks. When multiple description \cite{MDS} encoded,
the information symbols allow for  tiered reconstruction qualities
of the original source information at the sinks.

Sinks are sorted into $1\le l\le n$ clusters, each cluster
comprising $n_i$ $(i=1,2,\ldots,l)$ sinks. $n=\sum_{i=1}^{l}n_i$. A
sink in cluster $i$ is characterized by a tuple
$(z_i,c_i,\epsilon_i)$ $(i=1,2,\ldots,l)$. $z_i$ is a real {\it
constant} in $[0,1)$ indicating the fraction of input symbols that
sinks in cluster $i$ expect to recover. $z_i$ could be related to
the target distortion at the sinks. The two types of sink nodes are
distinguished by $c_i=\mathds{1}_{\{\mbox{cluster $i$ is
decoding}\}}$. $\epsilon_i$ is the erasure rate of the BEC channels
that link the source node to the sink nodes in cluster $i$.
Depending on the performance measure, sinks in the same cluster can
often be treated as one single sink because the tuples fully
characterize their decoding behavior in this broadcasting scenario.

%% file: Allerton09_theorybasis2.tex
The decoding process of LT codes starts with simply recovering the
input symbols connected to the received output symbols of degree-1.
This initial recovery induces a new set of output symbols of
degree-1. The decoding can continue in the same manner as long as
there are output symbols of induced degree-1. Such symbols
constitute what is known as the {\it ripple}. The decoding process
halts when the ripple becomes empty. In \cite{amin,karp} and
\cite{2nd_moment}, the expected size of the ripple throughout of the
decoding process is given as a function of the number of unrecovered
information symbols. We restate here the part of Theorem 2 in
\cite{2nd_moment} that concerns the expected size of the ripple.

Assume $w\cdot k$ output symbols have been collected and can be used
for decoding of an LT code, for some positive constant $w$. Let
$u\cdot k$ be the number of unrecovered information symbols, for a
constant $u\in[0,1)$. Let $r^{(k)}(u)$ be the expected size of the
ripple, normalized by $k$.

\newtheorem{theorem}{Theorem}
\begin{theorem}(Maatouk and Shokrollahi~\cite[Thm.~2]{2nd_moment})\label{thm:ripplesize}
If an LT code of $k$ information symbols has degree distribution
specified by the moment generating function $P^{(k)}(x)$ (see
(\ref{eq:mgf})), then
\begin{equation}
r^{(k)}(u)=wu\Bigl(P^{(k)\prime}(1-u)+\frac{1}{w}\ln
u\Bigr)+\mathcal{O}\Bigl(\frac{1}{k}\Bigr),
\end{equation}
where $P^{(k)\prime}(x)$ stands for the first derivative of
$P^{(k)}(x)$ with respect to $x$.
\end{theorem}

The original theorem in \cite{2nd_moment} is stated for the case
where the number of output symbols collected by the receiver is more
than the total number of information symbols, i.e., $w>1$. However,
the proof suggests that the theorem also holds for any constant
$w<1$.

Assume that $P^{(k)}(x)$ converges to $P(x)=\sum_{i\ge1}p_ix^i$ as
$k\rightarrow\infty$; then we have
\begin{equation}
r(u)=\lim_{k\rightarrow\infty}r^{(k)}(u)= u\bigl(w P'(1-u)+\ln
u\bigr). \label{aymp_recovery}
\end{equation}

In order for the decoding process to carry on until at least a
fraction $z$ of the information symbols could be recovered, the
ripple size has to be kept positive. If we use the expected value to
roughly estimate the ripple size, we should have
\begin{equation*}
r(u)=u\bigl(w P'(1-u)+\ln u\bigr)>0, \quad\forall u\in(1-z,1],
\end{equation*}
or equivalently,
\begin{equation}
w P'(1-u)+\ln u>0, \quad\forall u\in(1-z,1],
\label{eqn:recovery_condition}
\end{equation}
Inequality (\ref{eqn:recovery_condition}) provides a guideline for
the design of the degree distribution $P(x)$.

It is interesting to consider the implications of inequality
(\ref{eqn:recovery_condition}) on $w$ and $z$ relationship when the
degree distribution is $p_1=1$, that is, all output symbols are of
degree 1. Then (\ref{eqn:recovery_condition}) should tell us how
many (on the average) output symbols of degree 1 we need in order to
recover fraction $z$ of the information symbols. Note that when
$p_1=1$ we have $P(x)=x$ and $P'(x)=1$, and in turn from
(\ref{eqn:recovery_condition}), we have $w+\ln u>0$, $\forall
u\in(1-z,1]$. Thus, $w\ge-\ln(1-z)$, and consequently, the optimal
value of $w$ is $-\ln(1-z)$.

Note that we would get the same result if we tried to answer the
question about $w$ and $z$ by using the coupon collecting problem,
also known as the urns-and-balls problem. Throw a number of balls
into $k$ urns. Each ball is thrown independently and falls into each
urn with equal probability. What is the number of balls $N$ needed
for the number of urns containing at least one ball to reach $s$?
Note that $N$ is a random variable. It has been derived in
\cite[Ch.~2]{feller} (see also \cite{monograph}) that the expected
number of $N$ is
\begin{eqnarray*}
\mathbb{E}[N]&=&k\Bigl(\frac{1}{k}+\frac{1}{k-1}+\dots+\frac{1}{k-s+1}\Bigr)\\
&\gtrapprox& k\ln\frac{k}{k-s+1}=-k\ln\Bigl(1-\frac{s-1}{k}\Bigr) .
\end{eqnarray*}
Set $z=s/k$, the portion of urns possessing at least one ball. Then,
as $k\rightarrow \infty$, $\mathbb{E}[N]\rightarrow-k\ln(1-z)$.

Now, assume that the number of collected output symbols of the LT
code specified in Theorem \ref{thm:ripplesize} is $W\cdot k$, where
$W$ is a random variable with mean $\omega$, and denote the
normalized expected ripple size as $k\rightarrow\infty$ as $r_W(u)$,
then

\newtheorem{corollary}[theorem]{Corollary}
\begin{corollary}
\begin{equation}
r_W(u)=u\Bigl(\omega P^{\prime}(1-u)+\ln u\Bigr).
\label{eqn:randomchannel}
\end{equation}
\end{corollary}

\begin{proof}
This is due to the linearity of the expected ripple size in $W$ for
given $u$ and $P$.
\begin{eqnarray*}
r_W(u) &=&E\left[Wu\Bigl(P^{\prime}(1-u)+\frac{1}{W}\ln
u\Bigr)\right]\\
&=&u\Bigl(E[W] P^{\prime}(1-u)+\ln
u\Bigr)\\
&=&u\Bigl(\omega P^{\prime}(1-u)+\ln u\Bigr)
\end{eqnarray*}
\end{proof}

Then, from \ref{aymp_recovery}, we have the recovery condition for
random $W$ with mean $\omega$
\begin{equation}
\omega P'(1-u)+\ln u>0, \quad\forall u\in(1-z,1],
\label{eqn:rand_recovery_condition}
\end{equation}

In the next section, we shall use (\ref{eqn:recovery_condition}) to
formulate our optimization problems for LT code degree distribution
design.

%% file: Allerton09_optprob2.tex
Recall from Section \ref{sec:model} tuples $(z_i,c_i,\epsilon_i)$,
$i=1,2,\ldots,l$ are used to characterize the $l$ sink clusters in
the streaming network. Let $t_i\cdot k$ be the number of output
symbols transmitted by the source up till the time when the sinks in
cluster $i$ are able to recover their targeted fraction $z_i$ of the
input symbols. Then, the normalized number of symbols a sink in
cluster $i$ receives has mean $t_i(1-\epsilon)$.

If cluster $i$ is decoding($c_i=1$), then let $x=1-u$ in
(\ref{eqn:rand_recovery_condition}); we have
\begin{equation}
(1-\epsilon_i)t_iP'(x)+\ln(1-x)>0, \quad\forall x\in[0,z_i).
\label{eqn:rec_dec}
\end{equation}

A non-decoding user recovering information from a rateless coded
stream of degree distribution specified by $P(x)$ is equivalent to a
decoding user recovering information from a coded stream of degree
distribution specified by $P_0(x)=(1-P'(0))+P'(0)x$.

If cluster $i$ is non-decoding ($c_i=0$), then let $p_1=P'(0)$, the
fraction of degree-1 symbols and we have
\begin{equation}
(1-\epsilon_i)t_ip_1+\ln(1-x)>0, \quad\forall x\in[0,z_i).
\label{eqn:rec_nondec}
\end{equation}

The monotonicity and continuity of the $\ln$ function simplify
(\ref{eqn:rec_nondec}) to
\begin{equation}
(1-\epsilon_i)t_ip_1+\ln(1-z_i)\ge 0. \label{eqn:rec_nondec_sim}
\end{equation}

\paragraph{Min-Max Latency}
In the interest of the transmitting source, we wish to minimize the
transmission time that could guarantee the recovery of targeted
$(z_1,z_2,\ldots,z_l)$ fractions of input symbols by the $l$ sink
clusters. In addition, for broadcasting time-sensitive streaming
data, new data await to be transmitted after the transmission of an
older block of data is finished. Minimizing the maximum latency is
especially important for keeping the entire communications scheme in
pace.

This optimization problem could be expressed as follows: {\small
\begin{eqnarray}
&\mbox{min.}_{P}& \max_{i}\ t_i    \label{opt_gen_lat_fair}\\
&\mbox{s.t.}&t_i(1-\epsilon_i)P'(x)+\ln(1-x) > 0, \quad 0\le x\le z_i, \notag\\&&\quad\quad\mbox{if cluster $i$ is decoding},i=1,2\ldots,l,\notag\\
&&t_i(1-\epsilon_i)p_1+\ln(1-z_i)\ge0, \notag\\&&\quad\quad\mbox{if
cluster $i$ is non-decoding},i=1,2\ldots,l,\notag
\end{eqnarray}}
or equivalently, {\small
\begin{eqnarray}
&\mbox{min.}_{P,t_0}& t_0    \label{opt_gen_lat_fair0}\\
&\mbox{s.t.}&t_0(1-\epsilon_i)P'(x)+\ln(1-x) > 0, 0\le x\le z_i, \notag\\&&\quad\quad\mbox{if cluster $i$ is decoding},i=1,2\ldots,l,\notag\\
&&t_0(1-\epsilon_i)p_1+\ln(1-z_i)\ge0, \notag\\&&\quad\quad\mbox{if
cluster $i$ is non-decoding},i=1,2\ldots,l.\notag
\end{eqnarray}}

Let $t_0^*(z_1, z_2,\ldots,z_l)$ be the optimal solution to Problem
(\ref{opt_gen_lat_fair0}). Then the achievable information recovery
region for transmission of $t\cdot k$ output symbols is given by
\begin{eqnarray*}
Z(t)&=&\{(z_1,z_2,\ldots,z_l): \\
&&t_0^*(z_1,z_2,\ldots,z_l)\le t,\\&&z_i\in[0,1), i=1,2,\ldots,l\}.
\end{eqnarray*}

As we will see in the next section, optimization results show that,
when there are two decoding clusters in the network, one with
perfect link conditions and the other with erasure rate
$\epsilon=0.5$, after the source has transmitted $1.6k$ output
symbols, the cluster with worse channels can recover $63\%$ of the
input symbols in the mean time when the cluster with perfect
channels can recover $95\%$. If the source uses ideal soliton or
robust soliton distributions, however, the cluster with worse
channels may hardly recover anything until about $2k$ output symbols
have been transmitted. Similar results can be seen for cases where
there is one decoding cluster and a non-decoding cluster present in
the network.

\paragraph{Max-Min Channel Utilization}
 The Shannon capacity of the BEC link to sink
cluster $i$ is $(1-\epsilon_i)$ bits per channel use. The channel
utilization of a link to cluster $i$ is then
$v_i=\frac{z_i}{(1-\epsilon_i)t_i}$. We wish to maximize the minimum
channel utilization on all links, which is equivalent to minimizing
the inverse of the channel utilization. {\small
\begin{eqnarray}
&\mbox{min.}_{P}& \max_{i}\frac{t_i(1-\epsilon_i)}{z_i}    \label{opt_gen_uti_fair}\\
&\mbox{s.t.}&t_i(1-\epsilon_i)P'(x)+\ln(1-x)\ge0, 0\le x\le z_i, \notag\\&&\quad\quad\mbox{if cluster $i$ is decoding},i=1,2\ldots,l,\notag\\
&&t_i(1-\epsilon_i)p_1+\ln(1-z_i)\ge0,\notag\\&&\quad \quad\mbox{if
cluster $i$ is non-decoding},i=1,2\ldots,l,\notag
\end{eqnarray}}

or equivalently, {\small
\begin{eqnarray}
&\mbox{min.}_{P,v_0}& v_0    \label{opt_gen_uti_fair0}\\
&\mbox{s.t.}&v_0z_iP'(x)+\ln(1-x)\ge0, \quad 0\le x\le z_i, \notag\\&&\quad\quad\mbox{if cluster $i$ is decoding},i=1,2\ldots,l,\notag\\
&&v_0z_ip_1+\ln(1-z_i)\ge0, \notag\\&&\quad\quad\mbox{if cluster $i$
is non-decoding},i=1,2\ldots,l.\notag
\end{eqnarray}}

Maximizing the min channel utilization proves to be irrelevant to
the channel conditions, as may be inferred from the expression of
Problem (\ref{opt_gen_uti_fair0}). As we will see in the next
section, high minimum channel utilization could be achieved when the
decoding cluster has either a very low or a very high demand. The
increase in the demand of the non-decoding cluster, on the other
hand, always degrades channel utilization.

\paragraph{Max-Min Throughput}
The throughput at each sink cluster $i$ may be defined as
$\frac{z_i}{t_i}$. It is of interest to measure the objective
channel degradation regardless of channel capacity so as to provide
reference for service pricing of the broadcast application. We wish
to maximize the minimum throughput of all sink clusters. This is
equivalent to minimizing the maximum of the inverse of the
throughput. The optimization problem is therefore expressed as
Problem (\ref{opt_gen_throu_fair}):

{\small
\begin{eqnarray}
&\mbox{min.}_{P}& \max_{i}\frac{t_i}{z_i}    \label{opt_gen_throu_fair}\\
&\mbox{s.t.}&t_i(1-\epsilon_i)P'(x)+\ln(1-x)\ge0, 0\le x\le z_i, \notag\\&&\quad\quad\mbox{if cluster $i$ is decoding},i=1,2\ldots,l,\notag\\
&&t_i(1-\epsilon_i)p_1+\ln(1-z_i)\ge0,\notag\\&&\quad \quad\mbox{if
cluster $i$ is non-decoding},i=1,2\ldots,l,\notag
\end{eqnarray}}

or equivalently, {\small
\begin{eqnarray}
&\mbox{min.}_{P,w_0}& w_0    \label{opt_gen_throu_fair0}\\
&\mbox{s.t.}&w_0z_i(1-\epsilon_i)P'(x)+\ln(1-x)\ge0, \quad 0\le x\le z_i, \notag\\&&\quad\quad\mbox{if cluster $i$ is decoding},i=1,2\ldots,l,\notag\\
&&w_0z_i(1-\epsilon_i)p_1+\ln(1-z_i)\ge0,\notag\\&&\quad
\quad\mbox{if cluster $i$ is non-decoding},i=1,2\ldots,l.\notag
\end{eqnarray}}

\paragraph{Minimum Average Latency} We are also interested in minimizing the average
latency of all sinks. This is a natural measure of overall
performance.

\begin{eqnarray}
&\mbox{min.}_{P}& \frac{\sum_{i=1}^{l}n_it_i}{n}   \label{opt_gen_avg_lat}\\
&\mbox{s.t.}&t_i(1-\epsilon_i)P'(x)+\ln(1-x)\ge0, 0\le x\le z_i, \notag\\&&\quad\quad\mbox{if cluster $i$ is decoding},i=1,2\ldots,l,\notag\\
&&t_i(1-\epsilon_i)p_1+\ln(1-z_i)\ge0,\notag\\&&\quad\quad\mbox{if
cluster $i$ is non-decoding},i=1,2\ldots,l.\notag
\end{eqnarray}

Optimization results show that, when all channels are perfect and
half of the sinks are decoding, half non-decoding, the optimized
achievable average latency with one single broadcast data stream is
mostly worse than broadcasting on separate channels data streams
individually optimized for different sinks. Details are presented in
the Section \ref{sec:results}.

Since our objectives are the minimization of increasing functions of
the latencies, with arguments similar to Lemma 2 of \cite{sanghavi},
we can claim that there must exist optimal solutions to Problems
(\ref{opt_gen_lat_fair0}), (\ref{opt_gen_uti_fair0}),
(\ref{opt_gen_throu_fair0}) and (\ref{opt_gen_avg_lat}) with
polynomials $P(x)$ of degree no higher than
$d_{\max}=\lceil\frac{1}{1-\max_i\{z_i\}}\rceil-1$. This promises
the ready conversion of Problems (\ref{opt_gen_lat_fair0}),
(\ref{opt_gen_uti_fair0}) and (\ref{opt_gen_throu_fair0}) into
linear programming problems by the method proposed in
\cite{sanghavi}. Problem (\ref{opt_gen_avg_lat}) may be converted to
a series of linear programming problems for fixed $p_1\in[0,1]$ when
there are only two sink clusters in the network, one decoding and
the other non-decoding. To solve the linear programming problems
numerically, the parameter $x$ in the constraints is evaluated at
discrete points and lower bounds for the minimization problems with
constraints continuous in $x$ are obtained. In the next section we
will show in detail the optimization results of these problems.

%% file: Allerton09_results.tex
\subsection{Application to 2-Cluster Situations}

Now we apply our optimization problem to the case where only two
sink clusters with distinct tuple characteristics,
$(z_1,c_1,\epsilon_1)$ and $(z_2,c_2,\epsilon_2)$ exist. We deal
with:
(1) $c_1=c_2=1$, $\epsilon_1=0, \epsilon_2=0.5$: both clusters are
decoding, but with diverse channel conditions;
(2) one cluster is decoding while the other is not, with equal or
diverse channel qualities.

Figure \ref{fig:fairlatencyregion} shows the contour graphs of the
outer bounds of the min-max latency on the $z_1-z_2$ plane for four
typical cases.

\begin{figure}[htbp]%
\centering
\subfigure[]{\label{subfig:fairlata}\includegraphics[scale=0.45]{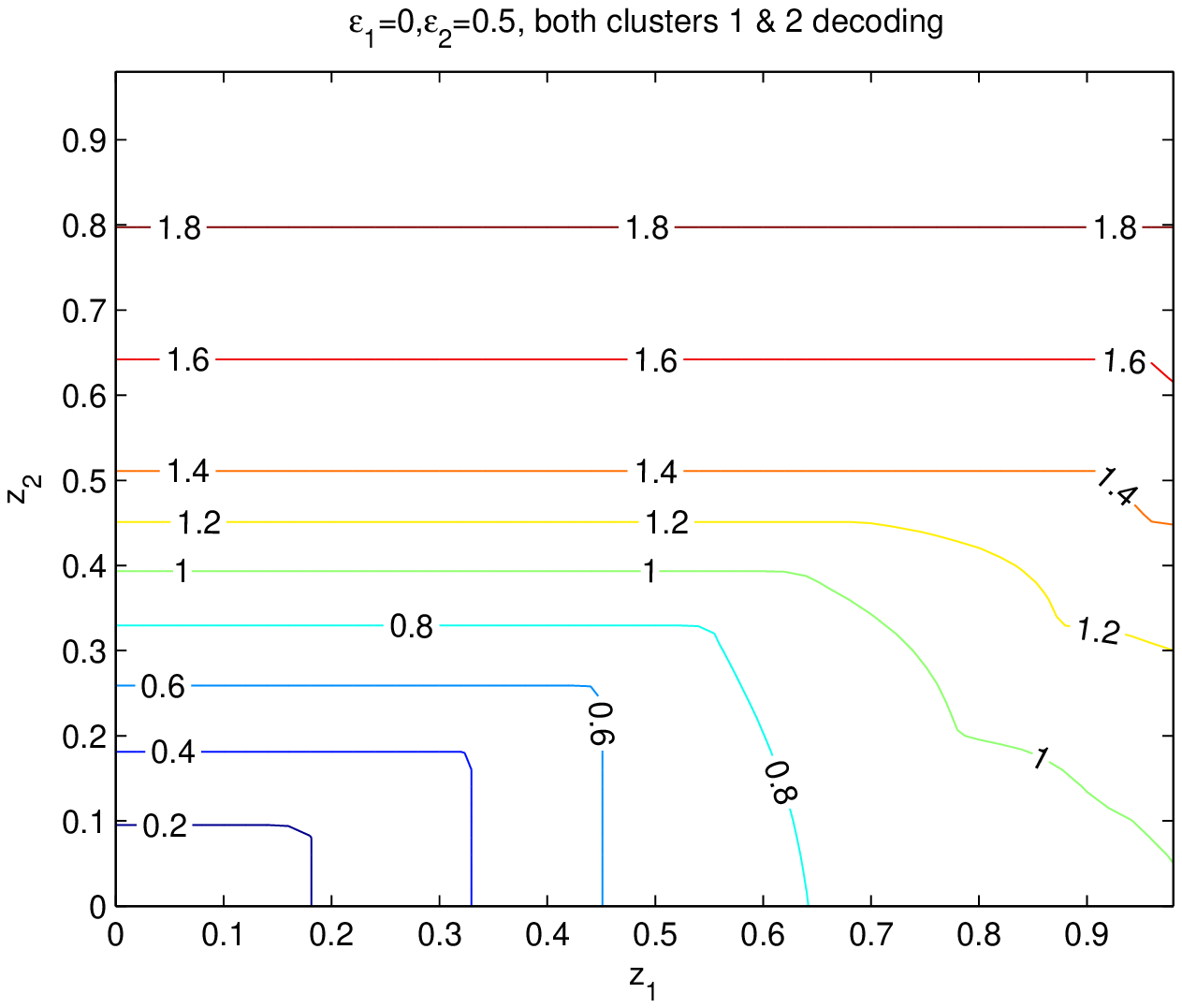}}\qquad
\subfigure[]{\label{subfig:fairlatb}\includegraphics[scale=0.45]{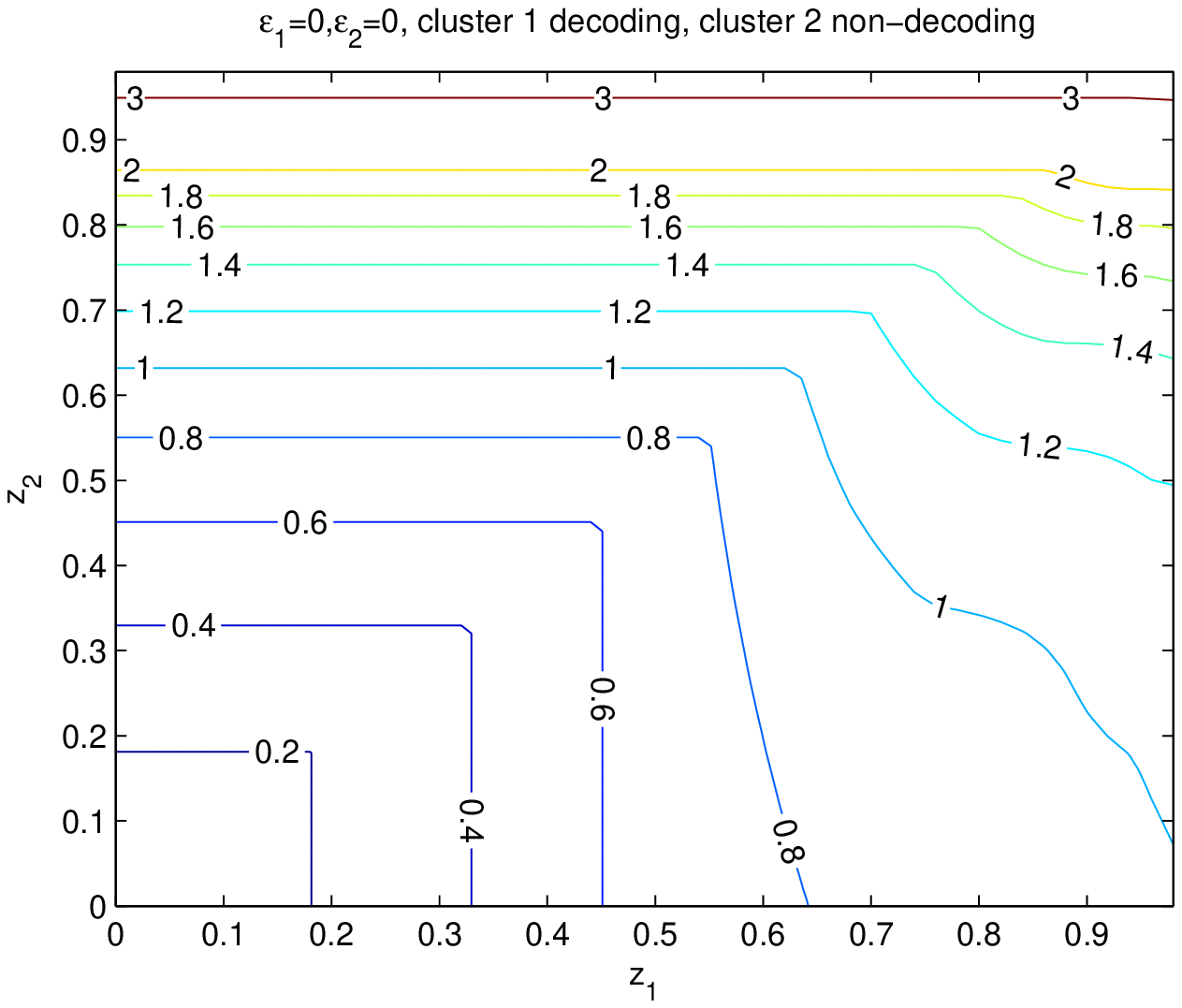}}\\
\subfigure[]{\label{subfig:fairlatd}\includegraphics[scale=0.45]{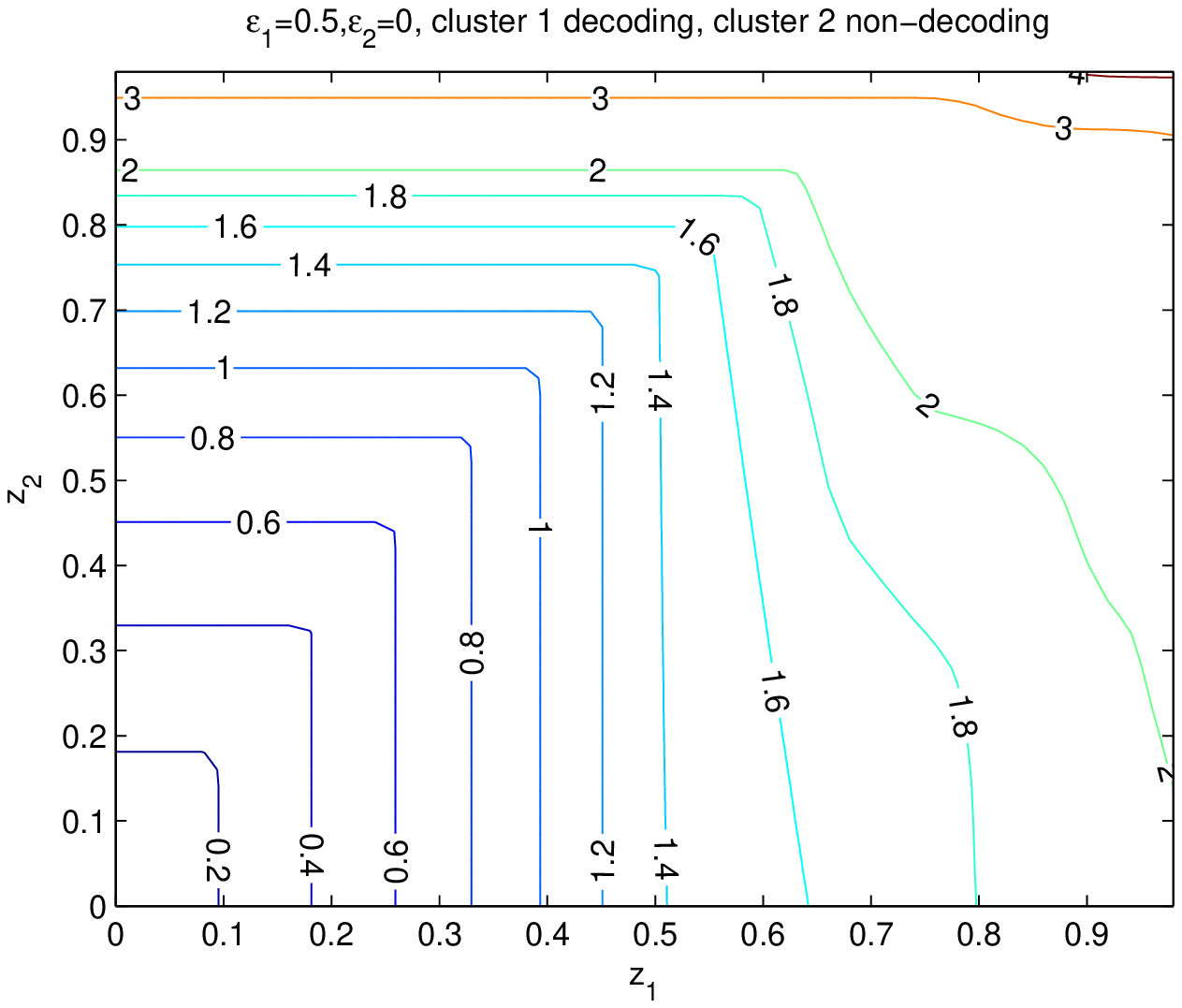}}
 \caption{Contour graphs of the
numerical lower bounds of the min-max latency. Contours define the
outer bounds of achievable $(z_1,z_2)$ regions given a specific
number of transmitted output symbols. Drawn from the solution to
Problem (\ref{opt_gen_lat_fair0}). } \label{fig:fairlatencyregion}
\end{figure}

\begin{itemize}
\item Dense contour regions indicate the regions where the minimized maximum latency is sensitive in $z_1$ or $z_2$;
\item Vertical(or horizontal) contour sections indicate regions where $z_1$(or $z_2$) is
the bottleneck of latency;
\item Steep(or gradual) contour sections indicate $z_1$(or $z_2$)-dominant regions:
reducing $z_1$(or $z_2$) a bit trades for a bigger advance in
$z_2$(or $z_1$) for fixed min-max latency. These are the regions
where the degree distribution of the LT encoder could be finely
tuned for the two clusters to finish reception at the same time.
\end{itemize}

Figures \ref{fig:fairutilregion}\subref{subfig:fairutia} and
\ref{fig:fairutilregion}\subref{subfig:fairutib} show respectively
the contour graphs of the outer bounds of the max-min channel
utilization when both sink clusters are decoding and when one
cluster is decoding but the other is not.

\begin{figure}[htbp]%
\centering
\subfigure[]{\label{subfig:fairutia}\includegraphics[scale=0.45]{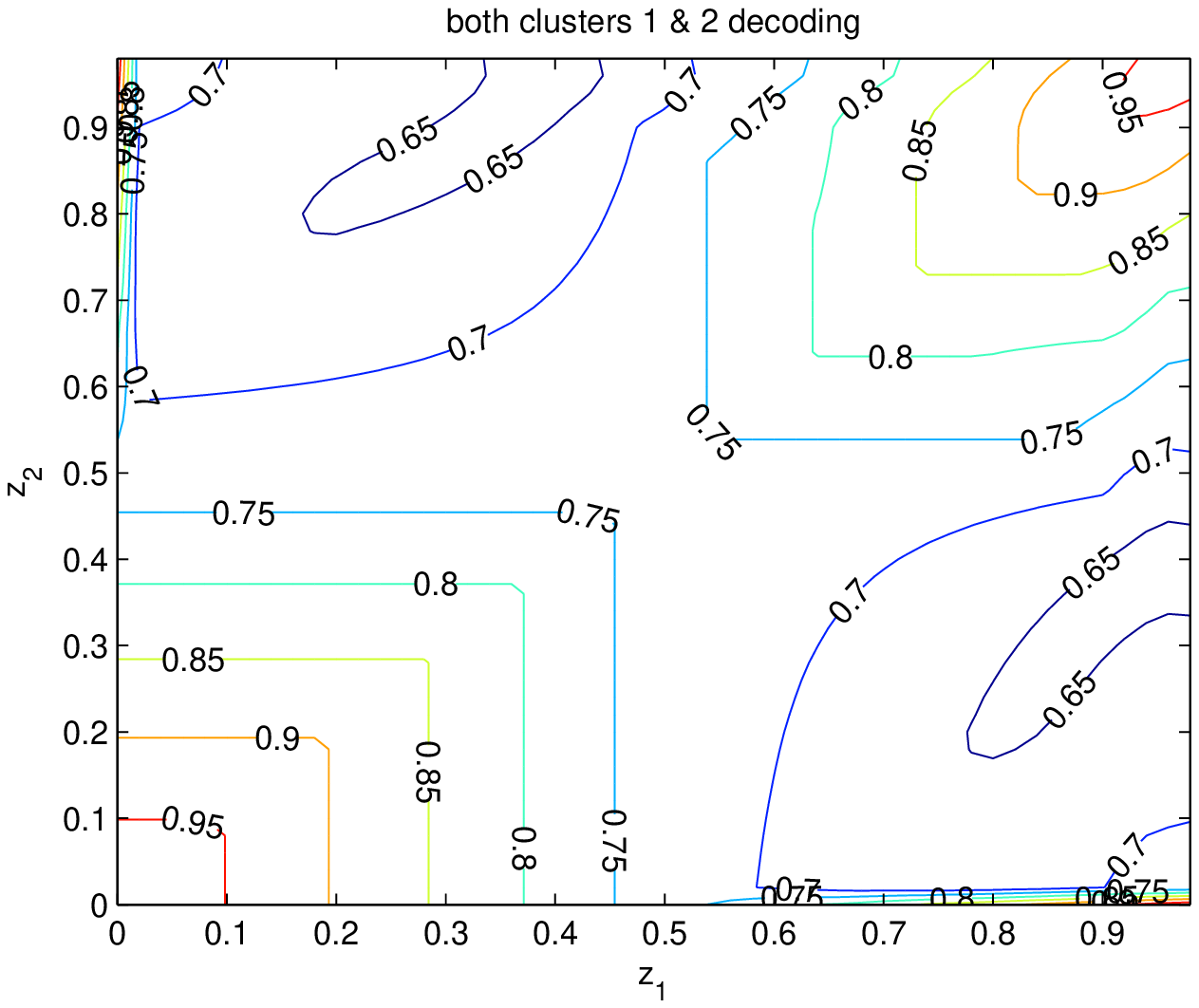}}\qquad
\subfigure[]{\label{subfig:fairutib}\includegraphics[scale=0.45]{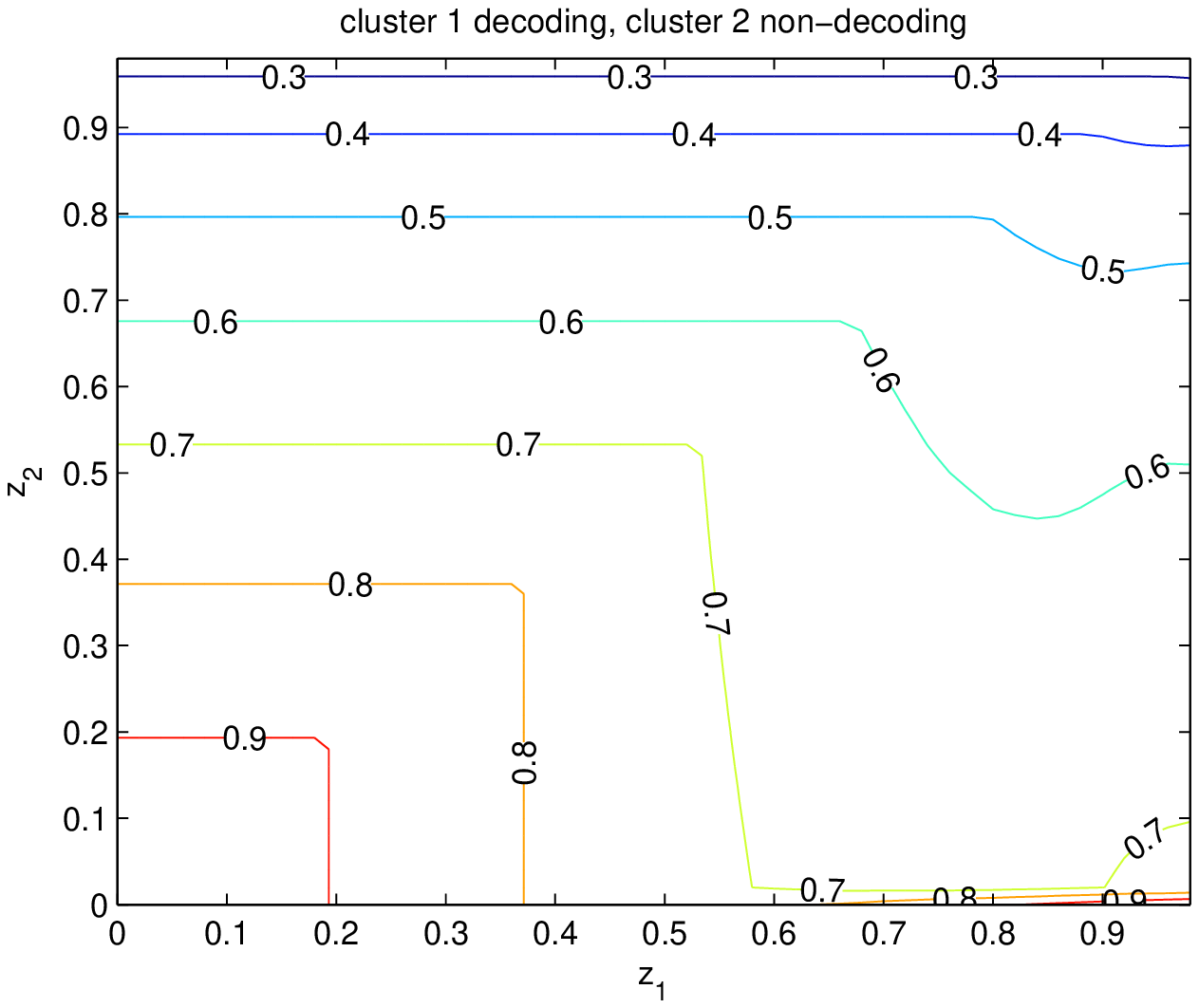}}
 \caption{Contour graphs of the
achievable max-min channel utilization. Drawn from the solution to
Problem (\ref{opt_gen_uti_fair}). } \label{fig:fairutilregion}
\end{figure}

\begin{itemize}
\item
The results are irrelevant to the channel quality;
\item Both clusters are decoding (Figure
\ref{fig:fairutilregion}\subref{subfig:fairutia}):
    \begin{itemize}
    \item For uniform demand $z=z_1=z_2$, channel utilization is the same as the slope of the outbound curve in the $z-r$ plot in \cite{sanghavi}: lowest as $z$ approaches $0.5$ and highest when $z$ is near $0$ or $1$.
    \item For non-uniform demand however diverse, the max-min channel utilization is better than $64\%$;
    \end{itemize}
\item Cluster 1 is decoding while cluster 2 is not(Figure
\ref{fig:fairutilregion}\subref{subfig:fairutib}):
    \begin{itemize}
    \item Max-min channel utilization decreases with increasing $z_2$;
    \item The ``lowest in the middle'' phenomenon could still be observed when $z_2$ is
    small;
    \item The minimum channel utilization could drop below $40\%$.
    \end{itemize}
\end{itemize}

For the results of maximizing the minimum throughput, we choose to
show the outbounds of the optimal solutions for $z_1=z_2$ under
different channel and decoding conditions in Figure
\ref{fig:fairthrouregion}.

\begin{figure}[htbp]%
\centering
\includegraphics[scale=0.4]{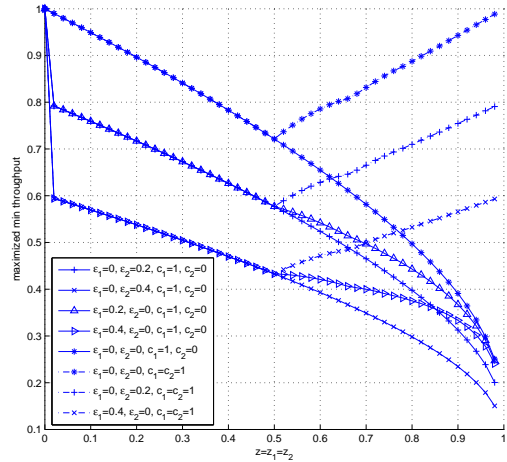}
 \caption{
Max-min
throughput versus $z=z_1=z_2$ under various channel conditions.
Drawn from the solution to Problem (\ref{opt_gen_throu_fair}). }
\label{fig:fairthrouregion}
\end{figure}


As shown in Figure
\ref{fig:fairthrouregion}
\begin{itemize}
\item The max-min throughput cannot go over the capacity of the worse channel, as expected;
\item The curves for both clusters decoding in different channel conditions are almost parallel and similar to the trend of the channel utilization,
which is also expected because of the uniform demand assumed here;
\item The curves for cluster 1 decoding and cluster 2 non-decoding
is always dropping with the growth of $z$; however, when the demand
is not uniform, when $z_2$ is small enough and $z_1$ large enough,
an increase in throughput could still be observed;
\item  The distance
between the outerbound max-min throughput curves for one cluster
decoding and the other not becomes smaller as $z=z_1=z_2$ grows
larger, which implies the less sensitivity of the optimized minimum
throughput to channel conditions when $z$ is larger.
\end{itemize}

Figure \ref{fig:avglat} shows the solution to Problem
(\ref{opt_gen_avg_lat}), minimizing the average latency.

\begin{figure}[htbp]%
\centering
\subfigure[Minimum achievable latency of the decoding
cluster vs. $p_1$ over perfect
channel]{\label{subfig:avgthroua}\includegraphics[scale=0.45]{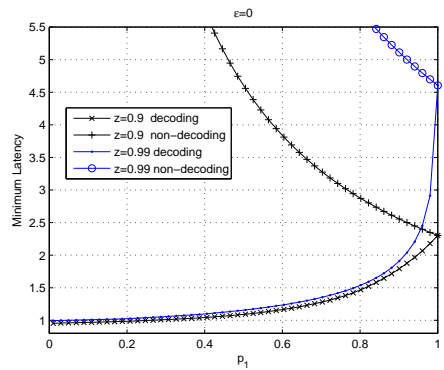}}\qquad
\subfigure[Minimum average latency and achieving $p_1$ versus size
of decoding cluster,
$z_1=z_2$.]{\label{subfig:avgthrouc}\includegraphics[scale=0.45]{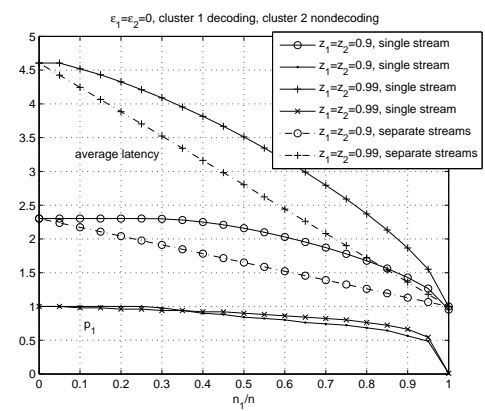}}\qquad
\subfigure[Minimum average latency contour
graph for half-half decoding-non-decoding]{\label{subfig:avgthroud}\includegraphics[scale=0.45]{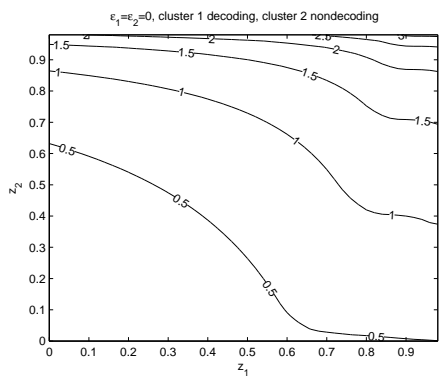}}\\
 \caption{Results for minimizing average latency, Problem
(\ref{opt_gen_avg_lat}). } \label{fig:avglat}
\end{figure}

\begin{itemize}
\item As shown in Figure \ref{fig:avglat}\subref{subfig:avgthroua}, on a perfect channel, even when half of the output symbols are of
degree-1, a decoding sink may be able to decode 99\% of all the
information symbols with an overhead of less than 16\% of the size
of the set of information symbols;
\item As shown in \ref{fig:avglat}\subref{subfig:avgthroud}, as the portion of decoding cluster increases from 0 to 1, the fraction of
degree-1 output symbols in the optimized degree distribution
gracefully decreases from 1 to 0.
\end{itemize}


\subsection{Comparison of Performance}
Table \ref{tab:lat_compare} lists a comparison of the total number
of transmitted symbols to fulfill the demands of two clusters under
four streaming schemes:
\begin{itemize}
\item Scheme A0: The source sends a single stream to all sinks,
minimizing max latency.
\item Scheme A1: The source sends a single stream to all sinks,
minimizing latency of cluster 1.
\item Scheme A2: The source sends a single stream to all sinks,
minimizing latency of cluster 2.
\item Scheme A12:  The source sends two independent streams to the
clusters, each minimizing latency of the targeted cluster.
\end{itemize}

\begin{table}[htbp]
 \caption{Comparison of Total Number of Transmitted
Symbols Under Four Streaming Schemes}
\begin{tabular}{|c|c|c|c|c|c|}
  \hline
   &$(z_1,c_1,\epsilon_1)$  & \multicolumn{2}{|c|}{Scheme A0}  & \multicolumn{2}{|c|}{Scheme A1}    \\
\cline{3-6}
   &$(z_2,c_2,\epsilon_2)$  & either & total & either & total \\
\hline
  cluster 1&(0.98,1,0)               & 1.5634 & \textbf{1.5634} & 0.9914  & $\mathbf{\infty}$  \\
  cluster 2&(0.72,0,0)               & 1.5634 &                 & $\infty$ &                  \\
  \hline
  cluster 1&(0.98,1,0)               & 1.6220 & \textbf{1.6220} & 0.9914 & \textbf{1.9828}  \\
  cluster 2&(0.63,1,0.5)               & 1.6220 &                 & 1.9828 &  \\
  \hline
   &$(z_1,c_1,\epsilon_1)$ & \multicolumn{2}{|c|}{Scheme A2}   & \multicolumn{2}{|c|}{Scheme A12}   \\
\hline
  cluster 1&(0.98,1,0)               & 3.9120 &\textbf{3.9120} & 0.9914 & \textbf{2.2644} \\
  cluster 2&(0.72,0,0)               & 1.2730 &               & 1.2730 &   \\
  \hline
  cluster 1&(0.98,1,0)               & 1.9959 & \textbf{1.9959} & 0.9914 & \textbf{2.5696} \\
  cluster 2&(0.63,1,0.5)             & 1.5782 &                    & 1.5782 &  \\
  \hline

\end{tabular}\label{tab:lat_compare}


\end{table}

Scheme A0 performs significantly better than Schemes A1, A2 and A12
in terms of the total number of output symbols transmitted by the
source.

When considering average latency for multicasting to both decoding
and non-decoding clusters on perfect channels, however, it could be
seen from Figure \ref{fig:avglat}\subref{subfig:avgthrouc} that
transmitting separately encoded streams on separate channels(Scheme
A12) is better than transmitting a single stream(Scheme A0).

%% file: Allerton09_simulation.tex
Figure \ref{fig:sim}\subref{subfig:simcurve} gives the simulated
sample curves of information recovery versus latency when the
decoding cluster targets at recovering 80\% of the input symbols and
the non-decoding cluster targets at recovering 40\%. The
distribution of the latency till the two clusters achieve targeted
information recovery is given in
\ref{fig:sim}\subref{subfig:simhist}. The empirical average value of
$t_0$ is $1.0718$, $2.3\%$ greater than the optimization result
$t_0^*=1.0473$, which is in acceptable error range.

\begin{figure}[htb]
\centering
\subfigure{\label{subfig:simcurve}\includegraphics[scale=0.55]{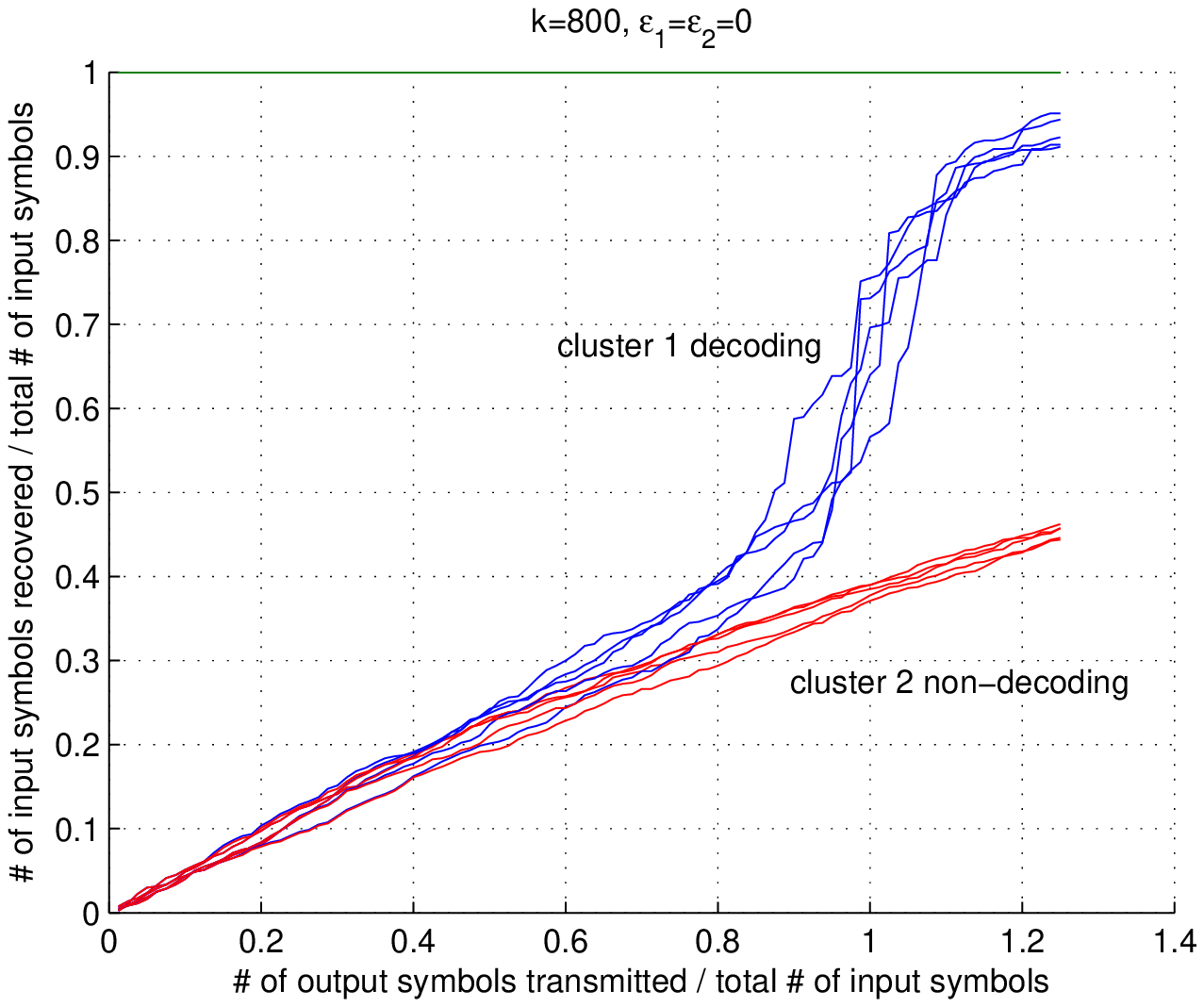}}\qquad
\subfigure{\label{subfig:simhist}\includegraphics[scale=0.55]{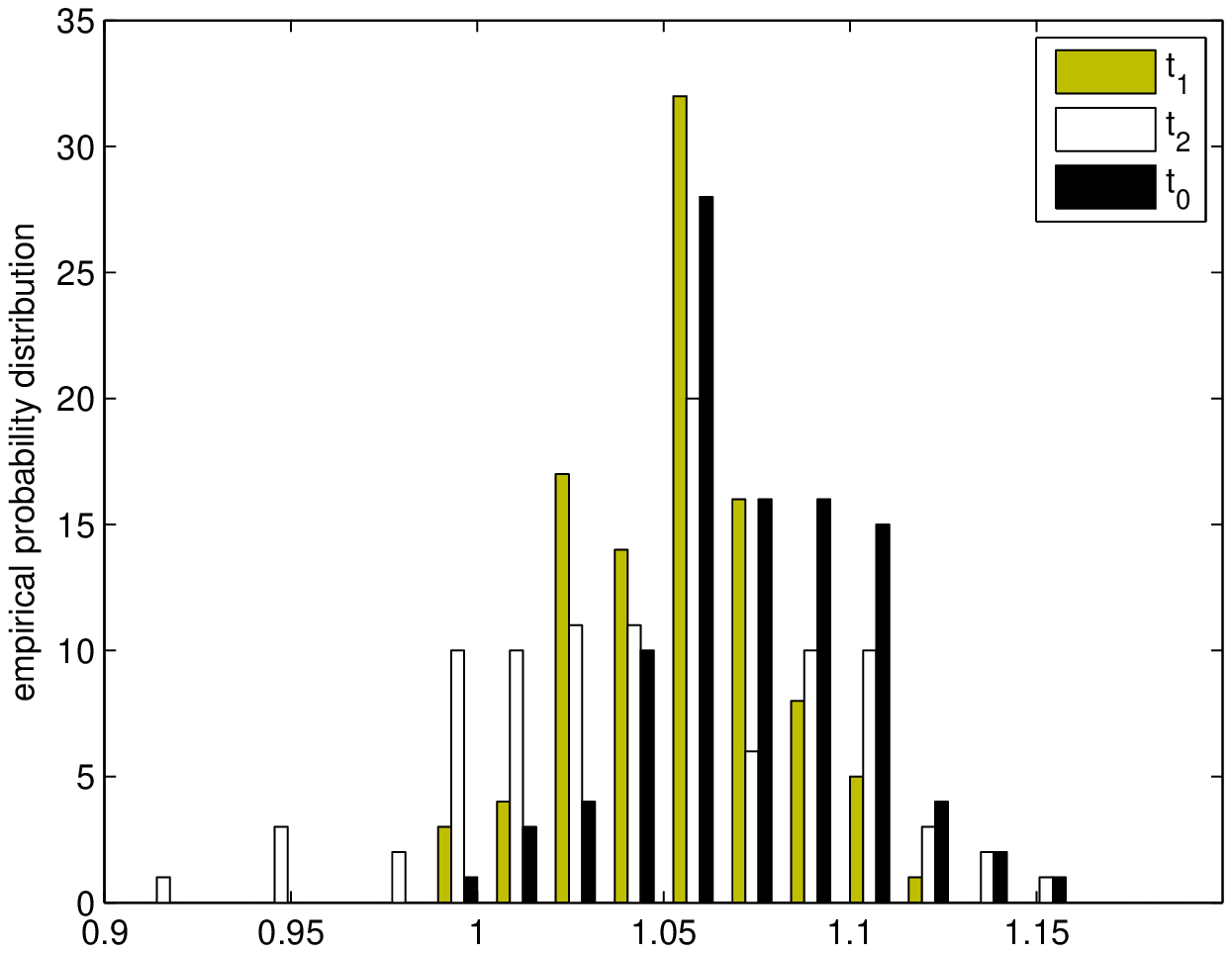}}
\caption{\subref{subfig:simcurve}Finite-length simulated time
progress of information recovery for degree distribution
$P(x)=0.4878x+0.4878x^4+0.0244x^5$, optimized for min-max latency
and $z_1=0.8$, $z_2=0.4$, $\epsilon_1=\epsilon_2=0$, the number of
information symbols being $k=800$. 5 simulation instances plotted.
\subref{subfig:simhist}Empirical probability distribution of latency
$t_1$ and $t_2$, obtained from 100 samples; mean of $t_1$ is
$1.0532$, standard deviation $0.0263$; mean of $t_2$ is $1.0451$,
standard deviation $0.0443$; mean of $t_0=\max\{t_1,t_2\}$ is
1.0718, standard deviation 0.0300. Optimization results give that
for $(z_1, z_2)=(0.8,0.4)$, min-max latency is $t_0^*=1.0473$.}
\label{fig:sim}
\end{figure}